\documentclass[12pt]{article}
\usepackage{amsmath, amssymb, amsthm, mathrsfs,amsfonts}
\usepackage{bbm}
\usepackage{capt-of}
\usepackage[bb=boondox]{mathalfa}
\usepackage{geometry}
\usepackage{graphicx}
\usepackage{hyperref}
\usepackage{mathtools}
\usepackage{tikz}
\usepackage{dsfont}
\usetikzlibrary{graphs}
\DeclarePairedDelimiter{\ceil}{\lceil}{\rceil}
\DeclarePairedDelimiter{\floor}{\lfloor}{\rfloor}

\providecommand{\keywords}[1]
{
  \small	
  \textbf{\textit{Keywords---}} #1
}
\theoremstyle{plain}
\newtheorem{theorem}{Theorem}[section]
\newtheorem{lemma}[theorem]{Lemma}
\newtheorem{corollary}[theorem]{Corollary}

\newtheorem{definition}[theorem]{Definition}
\newtheorem{observation}[theorem]{Observation}
\newenvironment{claim}[1]{\par\noindent\underline{Claim #1:}}{}
\newenvironment{claimproof}[1]{\par\noindent\underline{Proof of Claim:}\space#1}{\hfill $\blacksquare$}

\theoremstyle{definition}

\theoremstyle{remark}
\newtheorem{remark}[theorem]{Remark}
\newtheorem{question}[theorem]{Question}

\newcommand{\Z}{\mathbb{Z}} 

\def\gz {\Gamma(\mathbb{Z}_N)}

\def\gr {\Gamma(R)}
\def\gre {\Gamma_{E}(R)}

\def\dimc {dim_{COG}}
\def\dimt {dim_{TH}}

\def\t #1{\text{#1}}

\title{Boxicity of Zero Divisor Graphs}
\author{L. Sunil Chandran 
\\
\small Department of Computer Science and Automation, Indian Institute of Science\\
\small \texttt{sunil@iisc.ac.in}
\and 
Suraj Kumar Sahoo\\
\small Department of Computer Science and Automation, Indian Institute of Science\\
\small \texttt{surajks@iisc.ac.in}
}
\date{}

\begin{document}

\maketitle

\begin{abstract}

A $d$-dimensional box is the cartesian product $[a_1,b_1]\times \cdots [a_d,b_d]$ where each $[a_i,b_i]$ is a closed interval on the real line. The boxicity of a graph, denoted as $box(G)$, is the minimum integer $d\geq 0$ such that $G$ is the intersection graph of a collection of $d$-dimensional boxes.\\[2pt]
The study of graph classes associated with algebraic structures is a fascinating area where graph theory and algebra meet. A well known class of graphs associated with rings is the class of zero divisor graphs introduced by Beck in 1988. Since then, this graph class has been studied extensively by several researchers. Denote by $Z(R)$ the set of zero divisors of a ring $R$. The zero divisor graph $\gr$ for a ring $R $ is defined as the graph with the vertex set $V(\gr)=Z(R)$ and $E(\gr)=\{\{a_i,a_j\}: a_i,a_j\in Z(R)\text{ and } a_i a_j=0\}$.\\[2pt]
Let $N=\Pi_{i=1}^a p_i^{n_i}$ be the prime factorization of $N$. In {Discrete Applied Mathematics 365 (2025), pp. 260–269}, Kavaskar proved that $box(\gz)\leq \Pi_{i=1}^a(n_i+1)-\Pi_{i=1}^a(\floor{n_i/2}+1)-1$. 
In this paper we exactly determine the boxicity of the zero divisor graph $\gz$: 
 We show that when $N\equiv 2 \pmod 4$ and $N$ is not divisible by $p^3$ for any prime divisor $p$, we have $box(\gz)= a-1$. Otherwise $box(\gz)=a$. \\[2pt]
Suppose $R$ is a non-zero commutative ring with identity with the extra property that it is a reduced ring and let $k$ be the size of the set of minimal prime ideals of $R$. In the same paper, Kavaskar showed that $box(\gr)\leq 2^k-2$. We improve this result by showing $\floor{k/2}\leq box(\gr)\leq k$ with the same assumption on $R$. \\[2pt]
In this paper we also show that $a-1\leq \dimt (\gz)\leq a$ and $\floor{k/2}\leq \dimt(\gr)\leq k$, where $\dimt$ is another dimensional parameter associated with graphs known as the threshold dimension. 
\end{abstract}

\section{Introduction}

For a ring $R$, the set of zero divisors is defined as the set $$Z(R)=\{a\in R\setminus\{0\}|\exists b\in R\setminus\{0\} \t{  with  } ab=0\}$$
Beck in \cite{BECK1988208} associated a simple graph with any ring \( R \) by taking its elements as vertices and connecting two distinct vertices \( x \) and \( y \) if and only if $xy=0$ in \( R \). In this graph, the vertex representing \( 0 \) is adjacent to all other vertices. We can see that in such graphs, the vertices that correspond to elements of $R\setminus (Z(R)\cup \{0\})$ are only adjacent to $0$ and therefore they do not contribute to the understanding of the structure of the set of zero divisors. Consequently, D.F. Anderson et al.\cite{ANDERSON1999434} modified this concept by restricting the set of vertices to only the zero divisors of \( R \), maintaining the adjacency rule that \( x \) and \( y \) are adjacent if and only if \( xy = 0 \) in \( R \). This modified graph, known as the zero-divisor graph and denoted by \( \Gamma(R) \), has been studied extensively. \cite{9, ANDERSON20121626, 10,ANDERSON1999434,BECK1988208}.
\begin{definition}
     The zero divisor graph $\gr$ for a ring $R$ is defined as the graph with vertex set $V(\gr)=Z(R)$ and $E(\gr)=\{\{a_i,a_j\}| a_i,a_j\in Z(R)\text{ and } a_i a_j=0\}$.
\end{definition}

\noindent For a given family $\mathcal{F}$ of sets, the \textit{intersection graph} of $\mathcal{F}$ is the graph with $V(G)=\mathcal{F}$ and two vertices are adjacent if and only if the corresponding sets intersect. 
A $d$-dimensional box(or $d$-box) is a set $[a_1,b_1]\times \cdots [a_d,b_d]$ where $[a_i,b_i] $ are intervals on the real line. A $d$-dimensional cube is a $d$-box with the constraint $b_i-a_i=r$ for $1\leq i\leq d$ for some constant $r\geq 0$. A \textit{$d$-dimensional box (or cube) representation} for a graph $G$ is a function $f$ that maps every vertex $v\in G$ to a $d$-box ($d$-cube respectively) such that $uv\in E(G)$ if and only if $f(u)\cap f(v)\neq \emptyset$.   
\begin{definition}
    The boxicity of a graph $G$ $($denoted as $box(G)$$)$ is the minimum integer $d\geq0$ such that $G$ has a $d$-dimensional box representation. 
\end{definition}
\begin{definition}
    The cubicity of a graph $G$(denoted as $cub(G)$) is the minimum integer $d\geq0$ such that $G$ has a $d$-dimensional cube representation.
\end{definition}
By the above definitions, it follows that the boxicity and cubicity of complete graphs is $0$. The boxicity and cubicity of a graph are related in the following way:
\begin{theorem}[\cite{cub}]\label{cub}
    For a graph $G$ with $n$ vertices, $cub(G)\leq \ceil{\log_2(n)}box(G)$
\end{theorem}
\begin{definition}
     Interval graphs are the intersection graphs of the families of closed intervals on the real line. We will denote this class by \textit{INT}.
\end{definition}
We will denote an interval graph and its interval representation by the same symbol when there is no confusion. Suppose $I$ be an interval representation of an interval graph and $u$ be a vertex of the graph. We will denote by $I(u)$ the interval corresponding to $u$ in the interval representation $I$. 
 
\begin{theorem}[\cite{Roberts1969301}]\label{Roberts}
    A graph $G$ has boxicity at most $l>0$, if and only if it can be represented as the intersection of $l$ many interval graphs, i.e.$\t{ there exist interval supergraphs  } \{I_1, ... I_{l}\}$ of $G$ such that $\forall i\in [l]$, $V(I_i)=V(G)$ and $E(G)=E(I_1)\cap E(I_2)...\cap E(I_l)$.
\end{theorem}
\begin{definition}
    Let $G_1=(V_1,E_1), G_2=(V_2,E_2)$ be graphs. The join of $G_1$ and $G_2$, denoted as $G_1\vee G_2$ is the graph $G_3=(V_3,E_3)$ such that $V_3=V_1\sqcup V_2$ and $E_3=E_1\cup E_2\cup \{uv: u\in V_1, v\in V_2\}$. 
\end{definition}
\begin{definition}[Robert's Graph]
    Let $\overline{K_2}$, the complement of $K_2$, be the graph consisting of two isolated vertices. $\overline{K_2}\vee \overline{K_2} \cdots \vee \overline{K_2}$, the join of $t$ copies of $\overline{K_2}$ is isomorphic to the complement of $t$ isolated edges. We denote this graph as $\overline{tK_2}$ and it is known as Roberts' graph of order $2t$. 
\end{definition}
This graph was used by Roberts in his pioneering paper to demonstrate the existence of graphs on vertices with boxicity $\lfloor\frac{n}{2}\rfloor$.
The following two observations are not hard to show and will be used to prove lower bounds.
\begin{observation}
    Let $\overline{tK_2}$ be the Robert's graph of order $t$. Then $box(\overline{tK_2})=t$. 
\end{observation}
\begin{observation}\label{sub}
    Let $G$ be a graph and $H$ be an induced subgraph of $G$. Then $box(G)\geq box(H)$.
\end{observation}

\begin{definition}
    $G$ is a cograph if it can be constructed from isolated vertices by a sequence of disjoint union and join operations. We will denote this class of graphs as $COG$.
\end{definition}
 
 \begin{definition}
     Split graphs are the graphs in which the vertex set can be partitioned into a clique and an independent set. We will denote this class as $SPLIT$.
 \end{definition}
\begin{definition}
    The class of threshold graphs $TH=SPLIT\cap COG$.
\end{definition}
The following is an equivalent definition of threshold graphs. 
\begin{definition}\label{thresholddef}
    A graph $G$ is a threshold graph if there is a real number $S$ (the threshold) and for every vertex $v$ there is a real weight $a_v$ such that: $vw$ is an edge if and only if $a_v + a_w \geq  S$.
\end{definition}
A general notion of intersection dimension with respect to a graph class was introduced and studied in \cite{Kratochvil}.
\begin{definition}[Intersection Dimension of a Graph \cite{Kratochvil}]\label{defdim}
    Let $G$ be a graph and $e\in E(G)$. Let $\mathcal{A}$ be a class of graphs that contains $K_n$ and $K_n\setminus e$ $\forall n$. The intersection dimension of $G$ with respect to $\mathcal{A}$ (denoted as $dim_{\mathcal{A}}(G)$) is defined as follows:
    $$dim_{\mathcal{A}}(G)=\min\left\{k: \exists \t{ } G_1, ... G_k\in \mathcal{A} \t{ such that }E(G_1)\cap E(G_2)... \cap E(G_k)=E(G)\right\}$$
\end{definition}
We will study the case when $\mathcal{A}=TH$ and when $\mathcal{A}=COG$. $\dimt(G)$ is called the threshold dimension of $G$ and $\dimc(G)$ is called the cograph dimension of $G$. It can be shown that $TH\subset INT$. Therefore, it follows that for any graph $G$, $\dimt(G)\geq box(G)$.
\section{Our Results} We show the following:
\begin{itemize}
  \item Let $N=\Pi_{i=1}^a p_i^{n_i}$ be the prime factorization of $N$ with $a\geq 2$. In \cite{TK}, Kavaskar proved that $box(\gz)\leq \Pi_{i=1}^a(n_i+1)-\Pi_{i=1}^a(\floor{n_i/2}+1)-1$. 
Our result gives the exact value of the boxicity for $\gz$ in terms of the number of prime divisors of $N$: 
 We show that when $N\equiv 2 \pmod 4$ and $N$ is not divisible by $p^3$ for any prime divisor $p$, we have $box(\gz)= a-1$. Otherwise $box(\gz)=a$.\\ We also show that $a-1\leq \dimt (\gz)\leq a .$
    \item Denote the size of the set of minimal prime ideals of $R$ as $k$. In \cite{TK} it was shown that $box(\gr)\leq 2^k-2$ when $R$ is assumed to be a non-zero commutative ring with identity which is also a reduced ring.  We improve this by showing $$\floor{k/2}\leq box(\gr)\leq \dimt(\gr)\leq k$$ with the same assumption on $R$. Note that in \cite{TK}, they show that $box(\gr)\geq k$ but the proof given in the paper is incorrect because of an incorrect assumption(see remark \ref{err} for a discussion). Our lower bound uses similar ideas used in \cite{TK}, but avoids any assumption on $\gr$.
\end{itemize}  
\subsection{Results on \texorpdfstring{$\gz$}{gz}}
\begin{definition}
    For any prime $q$ we define the function $f(q,x)$ to be the non-negative integer such that $q^{f(q,x)}$ divides $x$ but $q^{f(q,x)+1}$ does not divide $x$.
\end{definition} 
It was shown in $\cite{TK}$ that when $N=p^n$ for $n\geq 3$, $box(\gz)=1$ and for $n\leq 2$, $box(\gz)=0$. Therefore, we only need to consider the case when $N$ has at least 2 prime divisors. 
     \begin{lemma}\label{lower}
     Let $N=\Pi_{i=1}^a p_i^{n_i}$ where $n_i>0$ and $a\geq 2$. Assume $p_j<p_k$ for $j<k$.
     \begin{enumerate}
         \item If $N\not\equiv2\pmod 4$ then $box(\Gamma(\mathbb{Z}_N))\geq  a$.
         \item  If $N\equiv 2\pmod 4$ and $n_j\geq 3$ for some $2\leq j\leq a$ then $box(\Gamma(\mathbb{Z}_N))\geq  a$.
         \item If $N\equiv 2\pmod 4 \t{ and } n_i\leq 2 \text{ for }2\leq i\leq a$ then $box(\Gamma(\mathbb{Z}_N))\geq  a-1$.
     \end{enumerate}
     \end{lemma}
     \begin{proof}
    We will show that the Robert's graph  $\overline{aK_2}$ is present as an induced subgraph in $\gz$ for cases $1$ and $2$ and the Robert's graph $\overline{(a-1)K_2}$ is present for case $3$. Then by observation \ref{sub} we will get our required lower bounds. 
    \begin{enumerate}
     \item [(1)] \underline{When $N\not\equiv2\pmod 4$:} \\[5pt]
     In this case, one of the following is true: (i)$N$ is not divisible by $2$ or (ii) $p_1=2$ and $n_1\geq 2$. \\ For every $1\leq i\leq a$, define $T_i:=\left\{\frac{N}{p_i^{n_i}},\frac{2N}{p_i^{n_i}}\right\}$. It is clear that $T_i\subset [N]$ since $\frac{2N}{p_i^{n_i}}< N$. Observe that $T_i$ is an independent set as an induced subgraph of $\gz$. Also, it is easy to see that for $i\neq j$, $T_i\cup T_j$ form the complete bipartite graph $K_{2,2}$. Therefore,  
     $$\t{}\bigcup_{i=1}^aT_i\t{ induces the Robert's graph $\overline{aK_2}$ in $\gz$} $$
     \item [(2)] \underline{When $N\equiv 2\pmod 4$ and $n_j\geq 3$ for some $2\leq j\leq a$:}\\[5pt]
     In this case, $p_1=2$ and $n_1=1$. Let $k$ be the integer such that $n_k\geq 3$. Define $T':=\left\{\frac{N}{2},\frac{2N}{2p_k}\right\}$ and $T'':=\left\{\frac{N}{p_k^{n_k-1}},\frac{2N}{p_k^{n_k-1}}\right\}$. It is clear that $T',T''$ and $T_i$(as defined above) are all subsets of $[N]$ as both $\frac{2N}{p_k^{n_k-1}}$ and $\frac{2N}{p_i^{n_i}}$ are smaller than $N$. $T'$ and $T_i$ are independent sets in $\gz$ as above. The fact that $T''$ is independent follows from observing that $\frac{N}{p_k^{n_k-1}}\cdot\frac{2N}{p_k^{n_k-1}}=\frac{N^2}{{p_k^{2n_k-2}}}$ contains $p_k^2$ as the highest power of $p_k$. Also, for $i\neq 1,k$ the sets $T'\cup T'', T'\cup T_i, T''\cup T_i$ form complete bipartite graphs $K_{2,2}$. Therefore, 
     $$\t{}T' \cup T'' \cup\bigcup_{\substack{ i=2\\i\neq j}}^aT_i\t{ induces the Robert's graph $\overline{aK_2}$ in $\gz$.}$$
     \item [(3)] \underline{When $N\equiv 2\pmod 4 \t{ and } n_i\leq 2 \text{ for }2\leq i\leq a$:}\\[5pt] 
     In this case, $p_1=2, n_1=1.$ Then,
     $$\bigcup_{i=2}^aT_i\t{ induces the Robert's graph $\overline{(a-1)K_2}$ in $\gz$.} $$
    
     \end{enumerate}
     \end{proof}
     We now provide upper bounds. 
     \begin{lemma}\label{thmgz}
    Let $N=\Pi_{i=1}^a p_i^{n_i}$ where $n_i>0$ and assume $p_j<p_k$ for $j<k$. Then  $$box(\Gamma(\mathbb{Z}_N))\leq \dimt({\gz})\leq a
   $$
\end{lemma}
\begin{proof}  Define $g(k):=\frac{k}{2N}$ $\forall k\in [N]$. \\[5pt] For $1\leq i\leq a$, $0\leq j\leq n_i$, let $F_{i,j}=\{u\in[N]: f(p_i,\t{gcd}(u,N))=j\}$. Note that $f(p_i,u)=f(p_i,\gcd(u,N))$ when $f(p_i,u)\leq n_i$. When $f(p_i,u)> n_i$, then $f(p_i,\gcd(u,N)
)={n_i}$. It is easy to see that $\bigsqcup_{j=0}^{n_i} F_{i,j}=[N]$.  \\[5pt]
    
    Define interval graphs $I_i$ on the vertex set $V(\gz)$ as follows: $$I_{i}(v)=
     \begin{cases}
         [j+g(v),j+g(v)] & \text{ for } v\in F_{i,j} \t{ and }j<\ceil{\frac{n_i}{2}}\\[5pt]
         [n_i-j,n_i] & \text{ for } v\in F_{i,j} \t{ and } j\geq \ceil{\frac{n_i}{2}} \\[5pt] 
     \end{cases}$$
     Note that $[j+g(v),j+g(v)]$ is a point interval.\\  Fix $1\leq i\leq a$. Take distinct numbers $u,v\in [N]$ and let $f(p_i,\t{gcd}(u,N))=j_u$ and $f(p_i,\t{gcd}(v,N))=j_v$. We make the following claims:\\
     \begin{claim}1\label{claim1}
         $j_u+g(u)\neq j_v+g(v)$ 
     \end{claim}\\
     \begin{claimproof}
         If $j_u> j_v$, then $j_u-j_v\geq 1>\frac{1}{2}\geq g(v)-g(u).$ The case $j_v>j_u$ is symmetrical.\\
         If $j_u=j_v$, then $j_u-j_v=0\neq g(v)-g(u)$ as $g(u)\neq g(v)$ by definition. 
     \end{claimproof}\\
     \begin{claim}2\label{claim2}
         $I_i(u)$ and $I_i(v)$ intersect if and only if $j_u+j_v\geq n_i$.
     \end{claim}\\
     \begin{claimproof}\textbf{`If' part of the claim:}
        Suppose $j_u+j_v\geq n_i$. At most one of $j_u$ or $j_v$ can be strictly smaller than $\ceil{\frac{n_i}{2}}$. \begin{itemize}
            \item { If $j_u\geq \ceil{\frac{n_i}{2}}$ and $j_v\geq \ceil{\frac{n_i}{2}}$:} In this case, $I(u)=[n_i-j_u,n_i]$ and $I(v)=[n_i-j_v,n_i]$. Both these intervals contain $n_i$ in their intersection. 
            \item {{Now without loss of generality, let $j_u<\ceil{\frac{n_i}{2}}$ and $j_v\geq \ceil{\frac{n_i}{2}}$:}} In this case, $I(u)=[j_u+g(u),j_u+g(u)]$ and $I(v)=[n_i-j_v,n_i]$. By assumption $j_u+j_v\geq n_i$. Therefore, $n_i\geq j_u+g(u)\geq n_i-j_v+g(u)\geq n_i-j_v$. So the intervals intersect.
        \end{itemize} This proves the `if' part. \\ 
        {\bf `Only if' part of the claim:} Suppose $j_u+j_v<n_i$. Then at least one of $j_u$ or $j_v$ is strictly smaller than $\ceil{\frac{n_i}{2}}$. \begin{itemize}
            \item {If $j_u< \ceil{\frac{n_i}{2}}$ and $j_v< \ceil{\frac{n_i}{2}}$:} In this case, $I(u)=[j_u+g(u),j_u+g(u)]$ and $I(v)=[j_v+g(v),j_v+g(v)]$. By claim 1, the intervals do not intersect.
            \item {Now, without loss of generality let $j_u<\ceil{\frac{n_i}{2}}$ and $j_v\geq \ceil{\frac{n_i}{2}}$:} In this case, $I(u)=[j_u+g(u),j_u+g(u)]$ and $I(v)=[n_i-j_v,n_i]$. By assumption $j_u+j_v\leq  n_i-1$. Therefore, $j_u+g(u)\leq n_i-j_v+g(u)-1< n_i-j_v$. So the intervals do not intersect.
        \end{itemize}
        This finishes the proof of claim 2. 
     \end{claimproof}
     \begin{remark}
         Recall the definition of threshold graphs (\ref{thresholddef}). We claim that each $I_i$ defined as above is a threshold graph. To see this, in definition \ref{thresholddef}, take $S=n_i$ and $a_u=j_u$ for each vertex $u$. By claim 2 $I_i$ satisfies the definition of threshold graphs.   
     \end{remark}
     By the above two claims, it can be shown that $E(\gz)=E(I_1)\cap E(I_2) \cdots \cap E(I_a)$. Indeed, if $u,v\in [N]$ are adjacent in $\gz$ then $N$ divides $uv$. Consequently, for each $1\leq i\leq a$, $p_i^{n_i}$ divides $uv$. This means that $f(p_i,\t{gcd}(u,N))+f(p_i,\t{gcd}(v,N))\geq n_i$ for all $1\leq i\leq a$. It follows by claim 2 that $u$ and $v$ are adjacent in $I_i$ for all $1\leq i\leq a$. Therefore, the threshold graphs $I_i$ are all supergraphs of $\gz$. On the other hand, if $u$ and $v$ are not adjacent in $\gz$, then $\exists$ $i$ such that $f(p_i,\t{gcd}(u,N))+f(p_i,\t{gcd}(v,N))<n_i$. Again, by claim 2, $u$ and $v$ are not adjacent in $I_i$. Hence, we have showed that $E(\gz)=E(I_1)\cap E(I_2) \cdots \cap E(I_a)$. The result follows. \\
     \end{proof}
     Note that the upper bound for boxicity given by the lemma \ref{thmgz} matches with the lower bounds given by \ref{lower} except for the case when $N\equiv 2\pmod4$ and all prime divisors $p_i\neq 2$ have $n_i\leq 2$. But now we show that in this case, the upper bound can be improved to match the corresponding lower bound.    
     \begin{lemma}\label{boximprov}
         Let $N= \Pi_{i=1}^{a} p_i^{n_i}$ be the prime factorization of $N$ with $p_j<p_k$ for $j<k$ where $n_i\leq 2 \text{ for }2\leq  i\leq a$, $p_1=2$ and $n_1=1$.  Then  $$box(\Gamma(\mathbb{Z}_N))\leq a-1$$
     \end{lemma}
     \begin{proof} For $2\leq i\leq a; $\hspace{5pt}$ 0\leq j\leq 2;$ $ 0\leq k\leq 1$, let $$F_{i,j,k}=\{u\in[N]: f(p_i,\t{gcd}(u,N))=j  \t{ and }f(2,\t{gcd}(u,N))=k\}.$$ Define $A=\{i\in [a]\setminus \{1\}: n_i=1\}$ and $B=\{i\in [a]\setminus \{1\}: n_i=2\}$. We are removing $1$ from the set $[a]$ because $p_1=2$ and we will not construct any interval graph corresponding to the prime factor $2$. Note that this forces the indexing of the interval graphs to begin from $2$ instead of $1$, i.e. $I_2, I_3, \cdots I_a.$ \\[5pt]
     Fix $2\leq i\leq a$. For distinct numbers $u,v\in \gz$, define $j_u=f(p_i,\t{gcd}(u,N))$ and $k_u=f(2,\t{gcd}(u,N))$. Similarly define $j_v$ and $k_v$. We will construct interval graphs $I_i$ on $V(\gz)$ as follows:
     \begin{enumerate}
         \item For $i\in A$, we construct an interval graph $I_i$ that satisfies the following: for every pair of distinct vertices $u,v\in V(\gz)$, we have $I_i(u)\cap I_i(v)\neq \emptyset$ if and only if one of the following is true: (a)$j_u+j_v\geq 1$ and $k_u+k_v\geq 1$ or (b) $u,v\in F_{i,1,0}$. 
         \item For $i\in B$, we construct an interval graph $I_i$ that satisfies the following: for every pair of distinct vertices $u,v\in V(\gz)$, we have $I_i(u)\cap I_i(v)\neq \emptyset$ if and only if one of the following is true: (a)$j_u+j_v\geq 2$ and $k_u+k_v\geq 1$ or (b) $u,v\in F_{i,2,0}$. 
     \end{enumerate} 
     \begin{center}
     \begin{tabular}[t]{|c|c|}
    \hline
        \begin{tikzpicture}[>=stealth,scale=1.5]
  \tikzstyle{filled} = [draw, circle, fill=black, minimum size=12pt, inner sep=0pt]
  \tikzstyle{unfilled} = [draw, circle, fill=white, minimum size=12pt, inner sep=0pt]

  \node[unfilled,label=above:{$F_i(0,0)$}] (v00) at (0,1) {};
  \node[unfilled,label=above:{$F_i(0,1)$}] (v01) at (1,1) {};
  \node[filled,label=below:{$F_i(1,0)$}] (v10) at (0,0) {};
  \node[filled,label=below:{$F_i(1,1)$}] (v11) at (1,0) {};

  \draw (v10) -- (v11);
  \draw (v10) -- (v01);
  \draw (v11) -- (v00);
  \draw (v11) -- (v01);
\end{tikzpicture}
 &
\begin{tikzpicture}[x=1cm,y=1cm]
\draw[line width=2pt] (6,3.125) -- (10,3.125) node[midway,below]{$F_i(1,1)$} ;
  
  \draw[line width=2pt] (8,3.5) --(10,3.5) node[right] {$F_i(1,0)$};

  \draw[line width=0.5pt,dashed] (8,3.875)  -- (9,3.875) node[midway,above] {$F_i(0,1)$};
  \draw[line width=0.5pt,dashed] (6,3.5)  -- (7,3.5) node[midway,above] {$F_i(0,0)$};
  \end{tikzpicture}
  \\[30pt]
 \hline 
\begin{tikzpicture}[>=stealth,scale=1.5]
  \tikzstyle{filled} = [draw, circle, fill=black, minimum size=12pt, inner sep=0pt]
  \tikzstyle{unfilled} = [draw, circle, fill=white, minimum size=12pt, inner sep=0pt]

  \node[filled,label=above:{$F_i(2,0)$}] (v20) at (0,1) {};
  \node[filled,label=above:{$F_i(1,1)$}] (v11) at (1,1) {};
  \node[unfilled,label=left:{$F_i(0,1)$}] (v01) at (-0.1,0.5) {};
  \node[unfilled,label=right:{$F_i(1,0)$}] (v10) at (1.1,0.5) {};
  \node[filled,label=below:{$F_i(2,1)$}] (v21) at (0.5,0) {};
  \node[unfilled,label=below:{$F_i(0,0)$}] (v00) at (1.8,0) {};

  \draw (v20) -- (v11);
  \draw (v20) -- (v01);
  \draw (v21) -- (v20);
  \draw (v21) -- (v11);
    \draw (v10) -- (v11);
  \draw (v21) -- (v01);
  \draw (v21) -- (v10);
  \draw (v21) -- (v00);
 
\end{tikzpicture}
 & 
 \hspace{20pt}
\begin{tikzpicture}

 \draw[line width=2pt] (5,2.125) -- (11,2.125) node[below]{{$F_i(2,1)$}};
  \draw[line width=2pt] (8.1,2.5) -- (11,2.5)node[right]{$F_i(1,1)$};
  \draw[line width=2pt] (7,1.75) -- (9,1.75)node[midway,below]{$F_i(2,0)$};

  \draw[line width=0.5pt,dashed] (5,2.5)  -- (6,2.5) node[midway,above]{$F_i(0,0$)};
   \draw[line width=0.5pt,dashed] (9.1,2.875)  -- (10.1,2.875)node[midway,above]{$F_i(1,0)$}; 
   \draw[line width=0.5pt,dashed] (7,2.5)  -- (8,2.5)node[midway,above]{$F_i(0,1)$};

\end{tikzpicture}
\hspace{20pt}
\\[30pt]
 \hline 
\end{tabular}
\end{center}

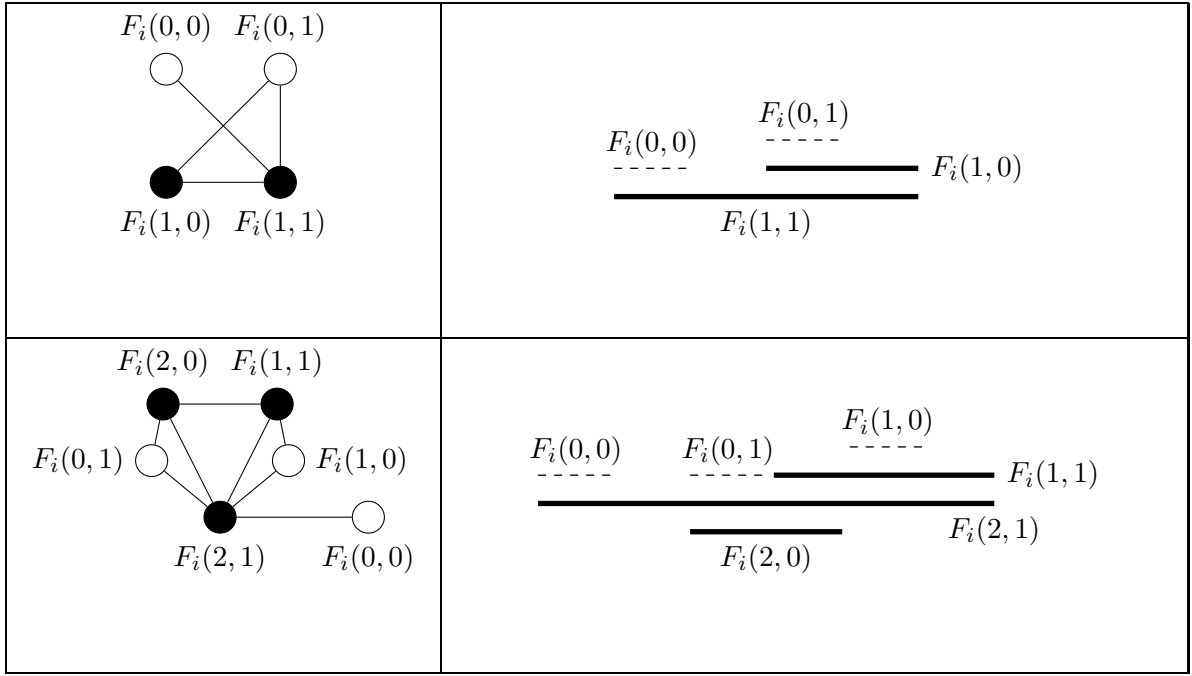
\captionof{figure}{\footnotesize{A schematic diagram of the graphs $I_i$ for $i\in A$(top) and $i\in B$(bottom). Each node represents a group of vertices $F_i(j,k)$. Edges are present between two nodes if and only if every vertex in one node is adjacent to every vertex in the other node. The hollow nodes indicate that the corresponding group of vertices form an independent set and solid nodes indicate cliques. The edges inside the cliques and the edges going between two nodes account for all the edges of the graph $I_i$. On the right, the corresponding intervals are illustrated for each vertex. The intervals corresponding to vertices of hollow nodes are drawn as small disjoint intervals with thin lines and the intervals corresponding to solid nodes are drawn long and thick. Note that vertices of a particular solid node is mapped to the same long interval(hence the use of thick lines).}}

    \vspace{10pt}
    These interval graphs are best described visually(see figure 1) but we also describe their mappings explicitly for completeness.\\
     [5pt] 
    \noindent 
    Define $g(k):=\frac{k}{2N}$.  \\[5pt]
    For $i\in A$ consider the interval graphs $$I_{i}(v)=
     \begin{cases}
         [g(v),g(v)] & \text{ for } v\in F_{i,0,0} \\[5pt]
         [1+g(v),1+g(v)] & \text{ for } v\in F_{i,0,1} \\[5pt]
         [0,2] & \text{ for } v\in F_{i,1,1} \\[5pt]
         [1,2] & \text{ for } v\in F_{i,1,0} \\[5pt]
     \end{cases}$$
     and for $i\in B$ consider the interval graphs $$I_{i}(v)=
     \begin{cases}
         [g(v),g(v)] & \text{ for } v\in F_{i,0,0} \\[5pt]
         [1+g(v),1+g(v)] & \text{ for } v\in F_{i,0,1} \\[5pt]
         [2+g(v),2+g(v)] & \text{ for } v\in F_{i,1,0} \\[5pt]
         [1,2] & \text{ for } v\in F_{i,2,0} \\[5pt]
         [1.5,3] & \text{ for } v\in F_{i,1,1} \\[5pt]
         [0,3] & \text{ for } v\in F_{i,2,1} \\[5pt]
     \end{cases}$$
      Note that the graphs so constructed are indexed as $I_2,I_3, \cdots, I_a$ as discussed above.
    
   It can be easily verified by inspection that these graphs satisfy the requirements $1$ and $2$.\\
We will show that $E(\gz)= E(I_2) \cdots \cap E(I_a)$. Indeed if distinct vertices $u,v\in V(\gz)$ are adjacent in $\gz$ then $N$ divides $uv$. Consequently, for each $1\leq i\leq a$, $p_i^{n_i}$ divides $uv$. It follows that $f(p_i,\gcd(u,N))+f(p_i,\t{gcd}(v,N))\geq n_i$ for every $1\leq i\leq a$. By conditions $1(a)$ and $2(a)$ above, $u$ and $v$ are adjacent in each $I_i$. Therefore, these graphs are supergraphs of $\gz$.\\
Now suppose $u,v\in V(\gz)$ be distinct vertices that are not adjacent. Then there is a prime factor $p_i$ of $N$ for which \( f(p_i,\t{gcd}(u,N))+ f(p_i,\t{gcd}(v,N))<n_i\). 
\begin{itemize}
    \item Suppose $p_i\neq 2$, i.e. $i\neq 1$. 
    \begin{itemize}
        \item If $i\in A$, then \( f(p_i,\t{gcd}(u,N))+ f(p_i,\t{gcd}(v,N))<1\), i.e. $j_u+j_v<1$. It follows that both $1(a)$ and $1(b)$ are violated for $I_i$. By construction, $u,v$ are not adjacent in $I_i$.
        \item If $i\in B$, then \( f(p_i,\t{gcd}(u,N))+ f(p_i,\t{gcd}(v,N))<2\), i.e. $j_u+j_v<2$. It follows that both $2(a)$ and $2(b)$ are violated for $I_i$. By construction, $u,v$ are not adjacent in $I_i$.
    \end{itemize}
    \item Suppose $p_i= 2$, i.e. $i=1$. Here $k_u+k_v<1$. Then condition $1(a)$ or $2(a)$ is violated for each interval graph $I_l$ depending on whether $l\in A$ or $l\in B$. If for any $l\in A$, $1(b)$ is violated by $u,v$ in $I_l$ then $u,v$ would not be adjacent in $I_l$. Similarly, if for any $l\in B$, $2(b)$ is violated by $u,v $ in $I_l$ then $u,v$ would not be adjacent in $I_l$. Therefore, we can assume that $u,v\in F_l(1,0)$ for all $l\in A$ and $u,v\in F_l(2,0)$ for all $l\in B$. This means that for $x=u,v$, we have $f(p_l,\gcd(x,N))=1$ when $l\in A$ and $f(p_l,\gcd(x,N))=2$ when $l\in B$. It follows that 
    $$u=k_1\cdot \prod_{l\in A}^ap_l\prod_{l'\in B}^ap_{l'}^2 \t{ and }v=k_2\cdot \prod_{l\in A}^ap_l\prod_{l'\in B}^ap_{l'}^2 \t{ for some $k_1,k_2\in \mathbb{Z}^+$}$$
    Recalling that $n_l=1$ for $l\in A$ and $n_l=2$ for $l\in B$, we get: $$u=k_1\cdot \prod_{l=2}^ap_l^{n_l} \t{ and }v=k_2\cdot \prod_{l=2}^ap_l^{n_l} \t{ where $k_1,k_2\in \mathbb{Z}^+$}.$$ Since $u$ and $v$ are smaller than $N$, we must have $k_1<2$ and $k_2<2$. This forces $k_1=k_2=1$ which contradicts the fact that $u$ and $v$ are distinct numbers. Hence, $\exists $ $l$ such that $u,v$ are not adjacent in $I_l$.
\end{itemize}
  
   It follows that $E(\gz)=E(I_2)\cap E(I_3) \cdots \cap E(I_a)$ as claimed. (Note that the index starts from $2$). Hence, by theorem \ref{Roberts}, $box(\gz)\leq a-1$. 
\end{proof}
The following theorem follows directly as a consequence of lemmas \ref{lower},\ref{thmgz}, \ref{boximprov}. 
\begin{theorem}\label{box}
    Suppose $N=\Pi_{i=1}^a p_i^{n_i}$, where $n_i>0$ and $a\geq 2$, be the prime factorization of $N$ with $p_1<p_2< \cdots < p_a$. Then, $$box(\Gamma(\mathbb{Z}_N))=\begin{cases}
        a-1 & \text{if } N\equiv 2\pmod 4 \t{ and } n_i\leq 2 \text{ for }2\leq i\leq a\\
        a & \text{otherwise }
    \end{cases}$$
\end{theorem}
 \begin{remark}
         In \cite{TK}, it was incorrectly shown that $box(\gz)=1$ if and only if $N=p^n$ for $n\geq 3$ and some prime $p$ or $N=2p$ for some odd prime $p$. Only one direction of this statement is true and an easy counterexample for the other direction is when $N=18$. It is not hard to verify that in this case $box(\gz)=1$ but $N$ is of the form $2p^2$ for $p=3$. Our result captures this fact. 
     \end{remark}
The following theorem follows from lemma \ref{thmgz} and the fact that for any graph $G$, we have \\$\dimt(G)\geq box(G)$.
\begin{theorem}\label{th}
    Suppose $N=\Pi_{i=1}^a p_i^{n_i}$ be the prime factorization of $N$. Then,  $$a-1\leq \dimt(\gz)\leq a$$
\end{theorem}

\begin{corollary}\label{cubi}
   Suppose $N=\Pi_{i=1}^a p_i^{n_i}$ be the prime factorization of $N$ with $p_1<p_2< \cdots < p_a$ and $a\geq 2$. Then $$\frac{\log (\lfloor \sqrt{N/p_1^2}\rfloor-1)}{2}\leq cub(\Gamma(\mathbb{Z}_N))\leq  a\left\lceil{\log N}\right\rceil$$
\end{corollary}
\begin{proof}

The upper bound follows from lemma \ref{thmgz} and theorem \ref{cub} and observing that $|V(\gz)|<N$.\\[2pt] For the lower bound, observe that $\alpha(\gz)\geq \lfloor \sqrt{N/p_1^2}\rfloor-1$ since for every pair of distinct positive\\[2pt] integers $u,v<\sqrt{N/p_1^2}$, we have $uv<N/p_1^2$. Therefore, the set $\{p_1u: u<\sqrt{N/p_1^2}\}$ forms an\\[2pt] independent set of size at least $\lfloor \sqrt{N/p_1^2}\rfloor-1$.
    It is folklore that for any graph $G$, $cub(G)\geq \frac{\log(\alpha(G))}{\log(d+1)}$,\\[2pt] where $\alpha(G)$ is the independence number of $G$ and $d$ is the diameter. In \cite{anderson2011zero} it was proved that for any \\[2pt] zero divisor graph, the diameter is less than or equal to 3.  The lower bound follows. 
\end{proof}
\begin{corollary}\label{corgzc}
    Let $N=\Pi_{i=1}^a p_i^{n_i}.$ Then, $$ \dimc(\Gamma(\mathbb{Z}_N))\leq a.$$
\end{corollary}
\begin{proof}
    The corollary follows from observing that threshold graphs are cographs and theorem \ref{th}.
\end{proof}
\subsection{Results on \texorpdfstring{$\gr$}{gr} for a reduced ring \texorpdfstring{$R$}{R}}
A reduced ring $R$ is a ring for which $x\in R$, $x^2=0\implies x=0$. In this section, we will consider the case when $R$ is a non-zero finite commutative ring with identity with the extra property that it is also a reduced ring. \\
 We can define an equivalence relation $\sim_E$ on any graph $G$ such that for $x,y\in V(G),$ we have $ x\sim_E y$ if and only if $N(x)=N(y)$. If $G$ is a simple graph then it has no self loops and so $\sim_E$ partitions $V(G)$ into equivalence classes $P=\{[x_1],[x_2],\cdots [x_k]\}$ for some $k$ such that $[x_i]$ induces an independent set in $G$ for every $1\leq i\leq k$.   \begin{definition}\label{redgrap}
    The reduced graph $G_E$ of a graph $G$ is defined as the graph with $V(G_E)=P$ and $E(G_E)=\{[x][y]\t{ }|\t{ }xy\in E(G)\}$. 
\end{definition} 
\noindent Following the above notation, we will denote the reduced graph of $\gr$ as $\Gamma_E(R)$.\\
\textit{Notational comment: Reduced rings should not be confused with the reduced zero divisor graphs of a ring.}\\ Since we will be working with reduced rings $R$, for every $x\in R$ we have $x^2\neq 0$ and therefore $\gr$ does not have self loops. Therefore, the equivalence classes of $\gr$ under $\sim_E$ are all independent sets. In fact, it is not hard to see that for distinct vertices $u,v\in V(\gr),$ $uv\in E(\gr)$ if and only if $[u][v]\in E(\gre)$. 

\begin{theorem}\label{thm:redring}
    Let $R$ be a non-zero finite commutative ring with identity which is also a reduced ring. Then $$\lfloor k/2\rfloor\leq box(\Gamma(R))\leq\dimt(\gr)\leq  k,$$ where $k$ is the number of minimal prime ideals of $R$. 
\end{theorem}
\begin{proof}

 From \cite{ANDERSON20121626} we know that there exists an isomorphism $\phi:\Gamma_E(R) \rightarrow \Gamma(R')$, where $R' = \mathbb{Z}_2^k = \mathbb{Z}_2 \times \mathbb{Z}_2 \times ... \times \mathbb{Z}_2$ ($k$-copies). Note that, $V(\Gamma(R'))=(\Pi_{i=1}^k\mathbb{Z}_2) \setminus \{\mathbb{1},\mathbb{0}\}$ where $\mathbb{1},\mathbb{0}$ are the all ones and all zeroes vectors. This isomorphism allows us to view $V(\gre)$ as vectors consisting of $0$s and $1$s such that distinct vertices $x,y\in V(\gre)$ are adjacent in $\gre$ if and only if the point-wise product of the corresponding vectors $\phi(x)$ and $\phi(y)$ equals the zero vector.  \\
    Let $\pi_i:\Gamma(R')\rightarrow\mathbb{Z}_2$ be the function that returns the $i$th coordinate. Using $\pi_i$, we can rewrite the above observation as follows: for vertices $x,y\in V(\gre)$,  $xy\notin E(\Gamma_E(R))$  if and only if $\exists$ $i$ such that $\pi_i\circ\phi(x)=\pi_i\circ\phi(y)=1$, i.e. for both the vectors the $i$th component is $1$.\\ Define the sets $A_{i}=\{x\in \gre: \pi_i\circ\phi(x)=1 \}$. Then the induced graph on the vertices of $A_i$ is an independent set in $\gre$ for each $1\leq i\leq k$. Thus by the above discussion $xy\notin E(\gre)$ if and only if $\exists $ $i$ such that $x,y\in A_i$\\ 
 Fix an injective function $g:\gr\rightarrow [0,1]$. For each $i$, define the graph $I_i$ on $V(\gr)$ as follows: 
 $$I_i(v)=\begin{cases}
    [g(v),g(v)] &\t{ if }  [v]\in A_i\\
    [0,1] & \t{otherwise }
 \end{cases}$$
  The following claim is easy to see:\vspace{5pt}
 \begin{claim}5
     For two distinct vertices $u,v\in V(\gr)$, $uv\notin E(I_i)$ if and only if both $[u],[v]\in A_i$.
 \end{claim} $\blacksquare$ \\[5pt]
 We will show that $E(\gr)=E(I_1)\cap E(I_2) \cdots \cap E(I_k)$.\\
 Fix $u,v\in \gr$. It is easy to verify the following chain of equivalences:
 \begin{align}
     uv\in E(\gr) &\iff [u][v]\in E(\gre)\\& \iff \phi([u])\phi([v])\in E(\Gamma(R'))\\& \iff \forall \t{ } 1\leq i\leq k,\t{ } \pi_i(\phi([u]))\cdot \pi_i(\phi([u]))=0\\& \iff \forall \t{ } 1\leq i\leq k, \t{ } \t{at least one of }[u] \t{ or }[v] \t{ belongs to } A_i^c.
 \end{align}
 By claim 5 and implication (4), it follows that if $uv\in E(\gr)$ then $uv\in E(I_i)$ $\forall$ $ 1\leq i\leq k$ and if $uv\notin E(\gr)$ then $\exists$ $i$ such that $uv\notin E(I_i)$. Therefore, $E(\gr)=E(I_1)\cap E(I_2) \cdots \cap E(I_k)$.\\ Each interval graph $I_i$ is also a threshold graph. To see this, assign the vertex $v$ the weight $a_v=0$ if $[v]\in A_i$ and $a_v=1$ otherwise. By claim 5, $u$ and $v$ are adjacent in $I_i$ if and only if $a_u+a_v\geq 1$. By definition \ref{thresholddef}, $I_i$ are threshold graphs. This proves the upper bound.\\ 
  Finally, suppose $e_j$ be the $j$th standard basis vector, i.e the vector with $1$ at the $j$th position and $0$ everywhere else. Define the set $B_i=\{e_{2i-1},e_{2i-1}+e_{2i}\}$ for $1 \leq i\leq \floor{k/2}$. \\ Observe that $\bigsqcup_{i=1}^{\floor{k/2}}B_i$ induces the Robert's graph $\overline{\floor{k/2}K_2}$ in $\Gamma(R')$. This is easy to see since $B_i$ induces an independent set for each $i\in [\floor{k/2}]$ and for $i\neq j$ all edges are present between $B_i$ and $B_j$. So $box(\gr)\geq box(\gre)=box(\Gamma(R'))\geq  \lfloor k/2\rfloor$. This gives the lower bound. 
\end{proof}
\begin{remark}\label{err}
    In \cite{TK}, it was shown that $box(\gr)\geq k$. The proof for the lower bound assumed a fact which is not always true. The proof goes as follows:\\ They begin with the fact that $\gre\cong\Gamma(R')$, where $R'\cong\Z_2\times\Z_2\cdots \Z_2$($k$ times). Recall that the vertices of $\gre$ are equivalence classes of vertices of $\gr$ under the equivalence relation $u\sim_Ev$ if $N(u)=N(v)$. Suppose the isomorphism between $\gre$ and $\Gamma(R')$ is given by $\phi$. Let $e_i$ be the vector consisting of $1$ at the $i$th position and $0$ elsewhere. For $1\leq i\leq k$, consider the subgraphs $G_i$ induced by the equivalence classes $[\phi^{-1}(e_i)]$ in $\gr$. Note that for $1\leq i\leq k$, $G_i$ is a graph with no edges, i.e. it is an independent set in $\gr$. The mistake in the proof comes from assuming that for these graphs $G_i$, $box(G_i)=1$. With this incorrect assumption, they argue that the graph induced by the set $V(G_1)\cup V(G_2) \cdots \cup V(G_k)$ has boxicity $k$. Since this graph exists as an induced subgraph of $\gr$, they argue that $box(\gr)\geq k$. \\
    We provide a counterexample for the assumption that $box(G_i)=1$ for all $1\leq i\leq k$. Observe that this is true only when $|V(G_i)|>1$. Consequently, it suffices to show that there exist graphs $\gr$ such that for some $1\leq i\leq k$, $|V(G_i)|=1$.\\ Take $ N=2\times \Pi_{i=2}^{k} p_i$ and $R=\Z_N$. Note that $R$ is a reduced graph as every $u\in \gz$ has at least one prime factor missing and so $u^2\not\equiv 0\pmod N$. The elements of $\gre$ are all the divisors of $N$(see \cite{TK} for a proof). Therefore the elements of $\gre$ are of the form $2^{\alpha_1}\times \Pi_{i=2}^{k} p_i^{\alpha_i}$ where $\alpha_i\in \{0,1\}$ for $1\leq i\leq k$. Define $\phi:\gre\rightarrow \Gamma(R')$ given by $\phi(2^{\alpha_1}\times \Pi_{i=2}^{k} p_i^{\alpha_i})=(1-\alpha_1, 1-\alpha_2, \cdots , 1-\alpha_k)$. It is easy to verify that the map $\phi$ is an isomorphism. Observe that $\phi^{-1}(e_i)=\frac{N}{p_i}$. Therefore, the equivalence class $[\phi^{-1}(e_i)]$ is the set $\left\{\frac{N}{p_i}, \frac{2N}{p_i}, \cdots , \frac{tN}{p_i}\right\}$ where $t=p_i-1$. In particular, the equivalence class $[\phi^{-1}(e_1)]=\left\{\frac{N}{2}\right\}$ is a singleton set. It follows that the graph $G_1$ induced by $[\phi^{-1}(e_1)]$ is an isolated vertex and hence $box(G_1)=0$. This contradicts that $box(G_i)=1$ for $1\leq i\leq k$.
\end{remark}

\begin{corollary}
        Let $R$ be a non-zero finite reduced commutative ring with identity. Then, $$\dimc(\Gamma(R))\leq k,$$ where $k$ is the number of minimal prime ideals of $R$.
\end{corollary}
\section{Conclusion} 

In this paper, we have exactly determined the boxicity of $\gz$. In corollary \ref{cubi} we have shown that $\frac{\log (\lfloor \sqrt{N/p_1^2}\rfloor-1)}{2}\leq cub(\Gamma(\mathbb{Z}_N))\leq  a\ceil{\log N}$. In \cite{adiga2009cubicity}, it was shown that for an interval graph $G$, we have\\[2pt] $cub(G)\leq \ceil{\log(\alpha(G))}$ where $\alpha(G)$ is the size of the largest independent set in $G$. Using this result and with a little more work, it can be shown that $$cub(\gz)\leq a \log N+ \sum_{j=1}^a\left\{\log\left(1-\prod_{i=1}^a\left(1-\frac{1}{p_i}\right)-\frac{1}{p_j^{\ceil{n_j/2}}}\right),\right\}$$ by finding the largest independent set of every interval graph defined in the proof of lemma \ref{thmgz}. We omit the proof since it is not a significant improvement. Note that the second term in the upper bound is negative. Similarly, the lower bound can be improved to $$cub(\gz)\geq \frac{1}{2}\log N+\frac{1}{2}\max_{j\in [a]}\log\left(1-\prod_{i=1}^a\left(1-\frac{1}{p_i}\right)-\frac{1}{p_j^{\ceil{n_j/2}}}\right)$$ by the same kind of argument as in corollary \ref{cubi}. It is an interesting open problem to find the exact cubicity of this graph. More precisely, we pose the following problem,
\begin{question}
    Is it possible to show $\Omega(\log N)\leq cub(\gz)\leq O(a+\log N).$
\end{question} 
In this paper, we have provided bounds for $box(\gr)$. When $R$ is a non-zero commutative ring with identity that is also a reduced ring and taking $k$ to be the size of the set of minimal prime ideals of $R$, we have showed that $\floor{\frac{k}{2}}\leq box(\gr)\leq \dimt(\gz)\leq k$. 
It is an interesting question to ask if the gap between the lower bound and the upper bound can be closed.  \\
We also wonder whether any lower bounds for cograph dimension can be obtained for both $\gz$ and $\gr$. 
\bibliographystyle{alpha}
\bibliography{references}

\end{document}